\documentclass{eptcs}
\usepackage{breakurl}  
\usepackage[utf8]{inputenc}
\usepackage[OT1]{fontenc}
\usepackage{graphicx}
\usepackage{vaucanson-g}
\usepackage{amsmath,amssymb,amsthm}
\newcommand{\ltype}{\vdash}
\newcommand{\rtype}{\dashv}
\newcommand{\lrtype}{{\ltype\!\!\!\rtype}}

\newcommand{\Ac}{\mathcal{A}}
\newcommand{\Bc}{\mathcal{B}}

\newcommand{\Sc}{\mathcal{C}}
\newcommand{\Dc}{\mathcal{D}}

\newcommand{\Nc}{\mathcal{N}}

\newcommand{\Eu}{\underline{E}}
\newcommand{\Iu}{\underline{I}}
\newcommand{\Tu}{\underline{T}}
\newcommand{\Fu}{\underline{F}}
\newcommand{\Ju}{\underline{J}}
\newcommand{\Uu}{\underline{U}}
\newcommand{\B}{\mathbb{B}}

\newcommand{\N}{\mathbb{N}}

\newcommand{\K}{\mathbb{K}}

\newcommand{\zero}{0_\K}
\newcommand{\one}{1_\K}

\newcommand{\supp}[1]{\underline{#1}}
\newcommand{\Tr}[3]{#1 \xrightarrow{\ #2\ } #3}
\newlength{\VariableStateHeight}
\newlength{\VariableStateITUDPos}
\newcommand{\StateVVar}[3][]%
{\StateStyle%
\settoheight{\VariableStateHeight}{\scalebox{\StateLabelSca}{\scalebox{\StateLabelScale}{$#1$}}}%
 \addtolength{\VariableStateHeight}{\ExtraSpace}
 \ifthenelse{\lengthtest{\VariableStateHeight < \VariableStateIntDiam}}%
       {\setlength{\VariableStateHeight}{\VariableStateIntDiam}}{}%
 \setlength{\VariableStateITUDPos}{\ArrowOnStateCoef\StateDiam}%
 \addtolength{\VariableStateITUDPos}{0.8\VariableStateHeight}%
 \addtolength{\VariableStateITUDPos}{-0.5\StateDiam}%
 \rput#2{\pnode(\ArrowOnStateCoef\StateDiam,0){#3e}%
         \pnode(-\ArrowOnStateCoef\StateDiam,0){#3w}%
         \pnode(0,\VariableStateITUDPos){#3n}%
         \pnode(0,-\VariableStateITUDPos){#3s}}%
\rput#2{\rnode{#3}{\psframebox[framearc=1]{\protect\rule[-.5\VariableStateHeight]{0pt}{1.7\VariableStateHeight}\protect\rule{\VariableStateIntDiam}{0pt}}}}
 \rput#2{\VaucStateRBLabel{#1}}%
}%
\newcommand{\StateDVar}[3][]%
{\StateStyle%
\settowidth{\VariableStateWidth}{\scalebox{\StateLabelSca}{\scalebox{\StateLabelScale}{$#1$}}}%
 \addtolength{\VariableStateWidth}{\ExtraSpace}
\settoheight{\VariableStateHeight}{\scalebox{\StateLabelSca}{\scalebox{\StateLabelScale}{$#1$}}}%
 \addtolength{\VariableStateHeight}{\ExtraSpace}
 \ifthenelse{\lengthtest{\VariableStateWidth < \VariableStateIntDiam}}%
       {\setlength{\VariableStateWidth}{\VariableStateIntDiam}}{}%
 \ifthenelse{\lengthtest{\VariableStateHeight < \VariableStateIntDiam}}%
       {\setlength{\VariableStateHeight}{\VariableStateIntDiam}}{}%
 \setlength{\VariableStateITPos}{\ArrowOnStateCoef\StateDiam}%
 \addtolength{\VariableStateITPos}{0.5\VariableStateWidth}%
 \addtolength{\VariableStateITPos}{-0.5\StateDiam}%
 \setlength{\VariableStateITUDPos}{\ArrowOnStateCoef\StateDiam}%
 \addtolength{\VariableStateITUDPos}{0.8\VariableStateHeight}%
 \addtolength{\VariableStateITUDPos}{-0.5\StateDiam}%
 \rput#2{\pnode(\VariableStateITPos,0){#3e}%
         \pnode(-\VariableStateITPos,0){#3w}%
         \pnode(0,\ArrowOnStateCoef\StateDiam){#3n}%
         \pnode(0,-\ArrowOnStateCoef\StateDiam){#3s}}%
\rput#2{\rnode{#3}{\psframebox[framearc=.7]{\protect\rule[-.5\VariableStateHeight]{0pt}{1.7\VariableStateHeight}\protect\rule{\VariableStateWidth}{0pt}}}}
 \rput#2{\VaucStateRBLabel{#1}}%
}%

\newtheorem{example}{Example}
\newtheorem{definition}{Definition}
\newtheorem{proposition}{Proposition}
\newtheorem{lemma}{Lemma}
\newtheorem{theorem}{Theorem}
\newtheorem{remark}{Remark}
\newtheorem{corollary}{Corollary}

\title{On Determinism and Unambiguity of Weighted Two-way Automata}

\author{Vincent Carnino
\institute{LIGM - Laboratoire d'informatique Gaspard-Monge\\Universit{\'e}
Paris-Est Marne-la-Vall\'ee, France}
\email{Vincent.Carnino@univ-mlv.fr}
\and
Sylvain Lombardy
\institute{LaBRI - Laboratoire Bordelais de Recherche en Informatique\\
Institut Polytechnique de Bordeaux, France}
\email{Sylvain.Lombardy@labri.fr}}

\def\titlerunning{On Determinism and Unambiguity of Weighted Two-way Automata}
\def\authorrunning{V. Carnino \& S. Lombardy}

\begin{document}
\maketitle

\begin{abstract}
In this paper, we first study the conversion of weighted two-way automata to one-way automata.
We show that this conversion preserves the unambiguity but does not preserve the determinism.
Yet, we prove that the conversion of an unambiguous weighted one-way automaton into a two-way automaton leads to
a deterministic two-way automaton.
As a consequence, we prove that unambiguous weighted two-way automata are equivalent to deterministic
weighted two-way automata in commutative semirings.
\end{abstract}

\section{Introduction}

A classical question in automata theory concerns the expressive power of a device and especially the difference between one-way devices and two-way devices. It is well known that two-way automata may be reduced to one-way automata and therefore recognize the same language family~\cite{Shepherdson59,RabinScott59}.

In this paper, we deal with the weighted versions of these two devices.
We describe the conversion of a two-way automaton over a commutative sering into a one-way automaton.
Such an algorithm has already be stated in~\cite{Anselmo90}; our construction is close, but we are mainly interested here in proving that this conversion preserves the unambiguity of automata; it does not preserve the determinism.

We then present a construction for the conversion of any unambiguous one-way automaton into a deterministic two-way automaton; this part does not require that the semiring is commutative.

A consequence of these two procedure is that, on commutative semirings, opposite to the case of one-way automata,
unambiguous two-way automata are not more powerful than deterministic ones.

\section{Weighted Two-way Automata}
\subsection{Automata and runs}
A semiring $\K$ is a set endowed with two binary associative operations, $\oplus$ and $\otimes$,
such that $\oplus$ is commutative and $\otimes$ distributes over $\oplus$. The set $\K$ contains two particular
elements, $\zero$ and $\one$ that are respectively neutral for $\oplus$ and $\otimes$; moreover, $\zero$ is an
annihilator for $\otimes$.

For every alphabet $A$, we assume that there exist two fresh symbols $\ltype$ and $\rtype$ that are marks
at the beginning and the end of the tapes of automata.
We denote $A_\lrtype$ the alphabet $A\cup\{\ltype,\rtype\}$. For every word $w$ in $A$,
$w_\lrtype$ is the word in $A_\lrtype$ equal to $\ltype w\rtype$.

One-way and two-way $\K$-automata share a part of their definition.
A $\K$-automaton is a tuple $\Ac=(Q,A,E,I,T)$ where
$Q$ is a finite set of states,
$A$ is a finite alphabet,
and $I$ and $T$ are partial functions from $Q$ to $\K$.
The support of $I$, $\Iu$, is the set of initial states of $\Ac$,
and the support of $T$, $\Tu$, is the set of final states of $\Ac$.

The definition of transitions differ.
In a two-way $\K$-automaton,
$E$ is a partial function from $Q\times (A_\lrtype \times \{-1,+1\}) \times Q$ into $\K$ and
the support of $E$, $\Eu$, is the set of transitions of $\Ac$.
Moreover, the intersection of $\Eu$ and $Q\times (\{\ltype\}\times \{-1\}\cup \{\rtype\}\times\{1\})\times Q$ must be empty.
\\Let $t$ be a transition in $\Eu$; if $t=(p,a,d,q)$, we denote $\sigma(t)=p$, $\tau(t)=q$, $\lambda(t)=a$, $\delta(t)=d$.
On figures, the value of $\delta$ is represented by a left (-1) or right (+1) arrow. For instance, if $t=(p,a,-1,q)$ and $E_t=k$, we draw
$\Tr{p}{a,\leftarrow|k}{q}$.

In a one-way $\K$-automaton,
$E$ is a partial function from $Q\times  A \times Q$ into $\K$,
and the support of $E$, $\Eu$, is the set of transitions of $\Ac$.
\\Let $t$ be a transition in $\Eu$; if $t=(p,a,q)$, we denote $\sigma(t)=p$, $\tau(t)=q$, $\lambda(t)=a$.

\begin{example}
Let $\Ac_1$ be the two-way $\mathcal{N}$-automaton of
Figure~\ref{fig:twoway}, where $\mathcal{N}=(\N\cup\{\infty\},\min,+)$ is the tropical semiring; since the multiplication law in this semiring is the usual sum, the weight of a path in this automaton is the sum of the weights of its transitions. This automaton is deterministic ({\it cf.} Definition~\ref{def:deterministic}) and thus there is only one computation for each accepted word. The behaviour of this automaton is quite easy. For each block of $'a'$ it checks through a left-right reading, whether the length of the block is odd; if it is, a right-left reading computes the length of the block; otherwise the automaton goes to the next block of $'a'$.
\end{example}
\begin{figure}
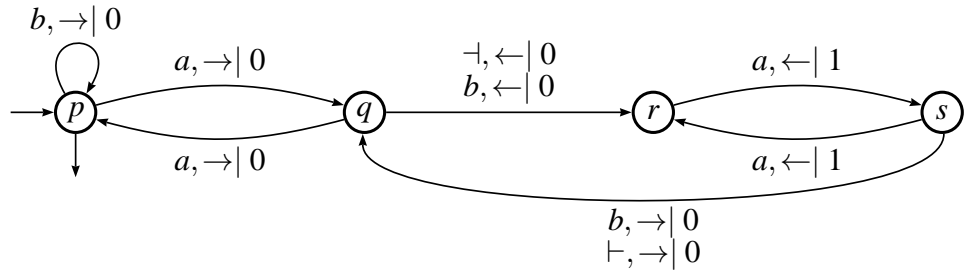
\FixVCScale{0.8}
\centering\VCCall{Dessins/twoway.tex}
\caption{The two-way $\mathcal{N}$-automaton $\Ac_1$.}
\label{fig:twoway}
\end{figure}

\begin{definition}
Let $w=w_1\dots w_n$ be a word of $A^*$, we set $w_0=\ \ltype$ and $w_{n+1}=\ \rtype$.
A \emph{configuration} of $\Ac$ on $w$ is a pair $(p,i)$ where $i$ is in $[0;n+1]$ and $p$ is a state of $\Ac$.
A computation (or run) $\rho$ of $\Ac$ on $w$ is a finite sequence of configurations $((p_0,i_0),\dots ,(p_k,i_k))$ such that :
\begin{itemize}
\item $i_0=1$, $i_k=n+1$, $p_0$ is in $\Iu$ and $p_k$ is in $\Tu$;
\item for every $j$ in $[0;k-1]$, there exists $t_j$, such that\\\hspace*{2cm} $\sigma(t_j)=p_j$, $\tau(t_j)=p_{j+1}$, $\lambda(t_j)=a_{i_j}$, and $i_{j+1}=i_j+\delta(t_j)$.
\end{itemize}
\end{definition}

The weight of such a computation, denoted by $|\rho|$, is $I(p_0)\otimes\bigotimes\limits_{j=0}^{k-1} E(t_j)\otimes  T(p_k)$. The weight of $w$ in $\Ac$, denoted by $\langle |\Ac|,w\rangle$, is the addition of the weights of all the runs with label $w$ in $\Ac$.
Notice that there may be an infinite number of computations with the same label $w$.
The definition of the behaviour of $\Ac$ in this case requires to study the definition of infinite sums. This can be done, like for one-way $\K$-automata with $\varepsilon$-transitions, for instance with complete semirings or topological semirings~\cite{LoSa13}.
This is not the purpose of this paper, since we mainly deal with two-way automata where the number of computations is finite for every word.

\begin{example}
A run of the $\Nc$-automaton~$\Ac_1$ over the word $abaaba$ is represented on Figure~\ref{fig:runabaaba}.
The weight of this run is equal to $2$.
\end{example}
\begin{figure}
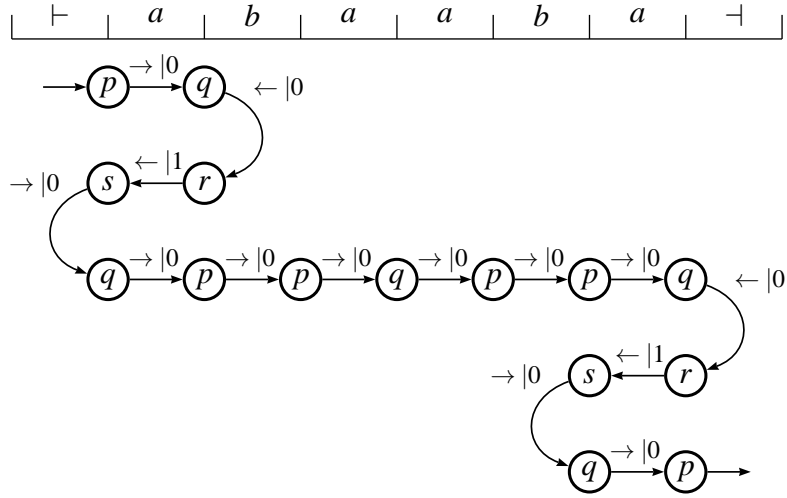
\FixVCScale{0.8}
\centering
\centering\VCCall{Dessins/runabaaba.tex}
\caption{A run of $\Ac_1$ over the word $abaaba$.}
\label{fig:runabaaba}
\end{figure}

\begin{definition}
Let $\rho=((p_0,i_0),\dots ,(p_k,i_k))$ be a run over $w$.
If there exists $m,n$ in $[1,k]$, with $m<n$ such that $(p_m,i_m)=(p_n,i_n)$,
then we say that $((p_m,i_m),\dots,(p_n,i_n))$ is an \emph{unmoving circuit} of $\rho$.
If $\rho$ does not contain any unmoving circuit, it is \emph{reduced}.
\end{definition}

\begin{lemma}\label{lem:red}
If a two-way $\K$-automaton admits a run $\rho$ which is not reduced,
it admits a reduced run with the same label.
\end{lemma}
\begin{proof}
We consider a shortest non reduced run $\rho=((p_0,i_0),\dots,(p_m,i_m),\dots,(p_n,i_n),\dots ,(t_k,i_k))$, with
$(p_m,i_m)=(p_n,i_n)$.
\\Then $((p_0,i_0),\dots ,(p_{m-1},i_{m-1}),(p_n,i_n),\dots,(p_k,i_k))$
is a run; by minimality of $\rho$, this run is reduced.
\end{proof}

\begin{definition}
A one-way or two-way automaton $\Ac$ is \emph{unambiguous} if every word labels at most one computation.
\end{definition}
Unambiguous automata have obviously only reduced computations.

\subsection{Coverings}
We extend here the notion of covering ({\it cf.} \cite{Saka09}) to two-way automata.
\begin{definition}
Let $\Ac=(Q,A,E,I,T)$ and $\Bc=(R,A,F,J,U)$ be two weighted two-way automata.
A mapping $\varphi$ from $Q$ into $R$ is a \emph{morphism} if,
\\i) $\forall p\in\Iu$, $J({\varphi(p)})=I(p)$;
\\ii) $\forall p\in\Tu$, $U({\varphi(p)})=T(p)$;
\\iii) $\forall t=(p,a,\delta,q)\in\Eu$, $\tilde\varphi(t)=(\varphi(p),a,\delta,\varphi(q))\in\Fu$ and $F({\tilde\varphi(t')})=E(t)$.
\\The morphism is surjective if $\varphi(Q)=R$, $\varphi(\Iu)=\Ju$, $\varphi(\Tu)=\Uu$, and $\tilde\varphi(\Eu)=\Fu$.
\end{definition}

\begin{definition}
Let $\Ac=(Q,A,E,I,T)$ and $\Bc=(R,A,F,J,U)$ be two weighted two-way automata.
$\Ac$ is  a \emph{covering} of $\Bc$ if there exists a surjective morphism $\varphi$ from
$\Ac$ onto $\Bc$ such that
\begin{gather*}
i)\ \forall r\in\Uu,\ \varphi^{-1}(r)\subseteq\Tu
\qquad
ii)\ \forall r\in\Ju,\ \exists!p\in\varphi^{-1}(r)\cap\Iu
 \\
iii)\ \forall t\in\Fu,\forall p\in\varphi^{-1}(\sigma(t)),
\exists!t'\in\tilde\varphi^{-1}(t), \sigma(t')=p.
\end{gather*}
$\Ac$ is  an \emph{in-covering} of $\Bc$ is there exists a surjective morphism $\varphi$ from
$\Ac$ onto $\Bc$ such that
\begin{gather*}
i)\ \forall r\in\Ju,\ \varphi^{-1}(r)\subseteq\Iu
\qquad
ii)\ \forall r\in\Uu,\ \exists!p\in\varphi^{-1}(r)\cap\Tu
 \\
iii)\ \forall t\in\Fu,\forall q\in\varphi^{-1}(\tau(t)),
\exists!t'\in\tilde\varphi^{-1}(t), \tau(t')=q.
\end{gather*}
\end{definition}

\begin{proposition}
Let $\Ac$ and $\Bc$ be two weighted two-way automata.
If $\Ac$ is a covering ({\it resp.} an in-covering) of $\Bc$, the corresponding morphism $\varphi$ induces a bijection between computations of $\Ac$ and $\Bc$ such that every computation of $\Ac$ and its image in $\Bc$ have the same label and the same weight.
\end{proposition}
\begin{proof}
Assume that $\Ac$ is a covering of $\Bc$.
Let $w$ be a word and let $((p_0,i_0),\ldots,(p_k,i_k))$ be a computation on $w$ in $\Ac$. For every $j$ in $[0;k]$, we set $r_j=\varphi(p_j)$;
by definition of a morphism $((r_0,i_0),\ldots,(r_k,i_k))$ is a computation on $w$ in $\Bc$ with the same weight.
Conversely, let $((r_0,i_0),\ldots,(r_k,i_k))$ be a computation in $\Bc$.
Let $p_0$ be the unique initial state in $\varphi^{-1}(r_0)$.For every $j$ in $[0;k-1]$, let $\delta_j=i_{j+1}-i_j$; the configuration $(r_{j+1},r_{j+1})$ is reached from configuration $(r_j,i_j)$ through the transition $(r_j,w_j,\delta_j,r_{j+1})$; inductively, we define $p_{j+1}$ as the unique state in $\varphi^{-1}(r_{j+1})$ such that $(p_j,w_j,\delta_j,p_{j+1})$ is a transition of $\Ac$. Then, $((p_0,i_0),\ldots,(p_k,i_k))$ is a computation on $w$ in $\Ac$. Hence, every computation $\rho$ of $\Bc$ is lift up in a unique way into a computation of $\Ac$ whose image by $\varphi$ is $\rho$.

The proof is similar for in-coverings.
\end{proof}
This proposition implies that a two-way automaton and its covering ({\it resp.} in-covering) are equivalent; moreover, if a two-way automaton is unambiguous, so is every of its (in-)coverings.

\subsection{$\delta$-Locality}
\begin{definition}
Let $\Ac$ be a two-way $\K$-automaton.
If, for each state $p$ of $\Ac$, every transition outgoing from $p$ has the same direction, then $\Ac$ is $\delta$-local.
\end{definition}

If $Q$ is the set of states of a two-way $\K$-automaton, we denote $Q_+$ ({\it resp.} $Q_-$) the set of states $p$
such that, for every transition $t$ outgoing from $p$, $\delta(t)=+1$ ({\it resp.} $\delta(t)=-1$); by convention, if $p$ has no outgoing transition, $p$ is in $Q_+$.
For every state $p$ of $Q_+$ ({\it resp.} $Q_-$), we set $\delta(p)=1$ ({\it resp.} $\delta(p)=-1$).

If $\Ac$ is a $\delta$-local automaton, $\{Q_+,Q_-\}$ is a partition of $Q$.

\begin{proposition}\label{prop:Ddet}
Every two-way $\K$-automaton admits a $\delta$-local in-covering.
\end{proposition}
\begin{proof}
In this proof, we denote $\pm=\{-1,+1\}$.
Let $\Ac = (R, A, F, J, U )$ be a two-way $\K$-automaton and let $P = R\setminus (R_+\cup R_-)$ be
the set of states in $\Ac$ such that there are at least two transitions with different direction outgoing
from each state. Let $P_+$ and $P_-$ be two copies of $P$ and
let $Q =R_+ \cup R_- \cup P_+ \cup P_-$. Let $\varphi$ be the canonical mapping from $Q$ onto $R$:
it maps every element of $P_+$ or $P_-$ onto the corresponding element of $P$.
Let $\tilde\varphi$ be the mapping from $Q \times A_\lrtype \times\pm\times Q$ into $R \times A_\lrtype \times\pm\times R$ defined by $\tilde\varphi(p, a, d, q) = (\varphi(p), a, d, \varphi(q))$.

Let $\Ac' = (Q, A, E, I, T )$ be the automaton defined by:
\begin{gather*}
\Iu = \varphi^{-1}(\Ju);\quad
\Tu =  \varphi^{-1}(\Uu)\setminus P_-;\\
\Eu = \{(p, a, d, q)\in\tilde\varphi^{-1}(\Fu) \mid (p,d)\in (P_+\cup R_+)\times\{+1\} \cup (P_-\cup R_-)\times\{-1\}
\};\\
\forall p\in\Iu,\ I(p)=J(\varphi(p)),\quad
\forall p\in\Tu,\ T(p)=U(\varphi(p)),\quad
\forall t\in\Eu,\ E(t)=F(\tilde\varphi(p)).
\end{gather*}
The automaton $\Ac'$ is $\delta$-local and it is an in-covering of $\Ac$.
\end{proof}

%
%
%
%
%
\begin{example}
The automaton $\Ac_1$ of Figure~\ref{fig:twoway} is not $\delta$-local;
from state $q$ ({\it resp. $s$}), there are transitions leaving with $\delta=1$ and other ones with $\delta=-1$.
The automaton $\Ac_1'$ of Figure~\ref{fig:twowayd} is a $\delta$-local in-covering of
$\Ac_1$. Notice that an in-covering of a deterministic automaton is not necessarily deterministic.

Actually, on $\Ac_1'$, transitions $\Tr{s_+}{b,\rightarrow|0}{q_-}$ and $\Tr{s_+}{\ltype,\rightarrow|0}{q_-}$
do not belong to any computation, since the label of any trnasition that would follow one of these boths transitions should be the same as the label of the transition arriving at $s_-$ ($a$), and there is no transition outgoing from $q_-$ with label $a$.
\begin{figure}
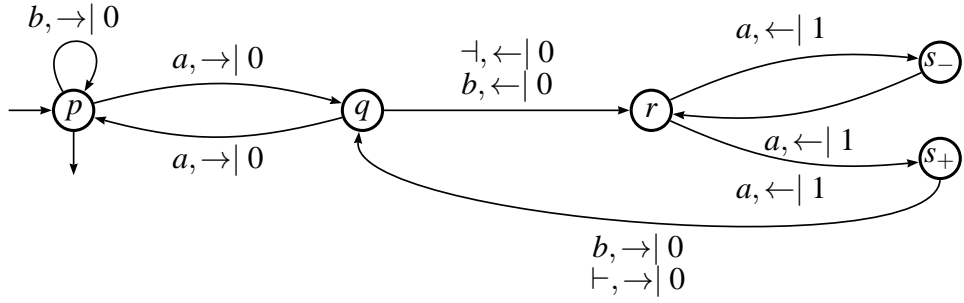
\FixVCScale{0.8}
\centering\VCCall{Dessins/twowayd.tex}
\caption{The $\delta$-local two-way distance automaton $\Ac_1'$.}
\label{fig:twowayd}
\end{figure}

\end{example}

\section{Slices}
In this section, we describe the conversion of two-way automata over commutative semirings into one-way automata.
We give sufficient conditions to get finite one-way automata.

\subsection{The Slice Automaton}
\begin{definition}
Let $\Ac=(Q,A,E,I,T)$ be a two-way $\K$-automaton and let $w=w_1\dots w_k$ be a word.
$\rho=((p_0,i_0),\dots (p_n,i_n))$ be a run over $w$, and $j$ in $[1;k+1]$.
Let $h$ be the  subsequence of all pairs $(p_r,i_r)$ such that
$(i_r,i_{r+1})=(j,j+1)$ or $(i_{r-1},i_r)=(j,j-1)$.
The $j$-th \emph{slice} of $\rho$ is the vector $s^{(j)}$ of states obtained by the projection
of the first component of each pair of $h$.
\\The \emph{signature} $S(\rho)$ of $\rho$ is the sequence of its slices.
\end{definition}
The slices we define here are not exactly the \emph{crossing sequences} defined in~\cite{Shepherdson59}.

\begin{example}
The vector $\left[\StackThreeLabels qrp\right]$
is the second (and the seventh) slice of the run of Figure~\ref{fig:runabaaba}.
The signature of this run is:
\begin{equation}
\left(
\StackThreeLabels psq,
\StackThreeLabels qrp,
\StackThreeLabels p{}{},
\StackThreeLabels q{}{},
\StackThreeLabels p{}{},
\StackThreeLabels psq,
\StackThreeLabels qrp
\right).
\end{equation}
The signature of the (unique) run on the word $abaaba$ in the automaton $\Ac_1'$
is
\begin{equation}
\left(
\StackThreeLabels p{s_+}{q_+},
\StackThreeLabels {q_-}rp,
\StackThreeLabels p{}{},
\StackThreeLabels {q_+}{}{},
\StackThreeLabels p{}{},
\StackThreeLabels p{s_+}{q_+},
\StackThreeLabels {q_-}rp
\right).
\end{equation}
\end{example}

Let $\Ac=(Q,A,E,I,T)$ be a $\delta$-local two-way $\K$-automaton.
To define a one-way $\K$-automaton from slices we consider the set $X$ of subvectors of slices,
that are vectors $v$ in $Q^*$ with an odd length; let $Y$ be the vectors $v$ in $Q^*$ with an even length.

We define inductively two partial functions $\theta:X\times A\times X\rightarrow \K$ and $\eta:Y\times A\times Y\rightarrow \K$ by:
\begin{equation}
\begin{split}
\eta(\varepsilon, a,\varepsilon) &= 0_\K,\\
\forall p,q\in Q,\qquad
\delta(p)=1\Longrightarrow\forall u,v\in Y,\  \theta(pu,a,qv)&=E(p,a,1,q)+\eta(u,a,v),\\
\eta(u,a,pqv)&=E(p,a,1,q)+\eta(u,a,v),\\
\delta(p)=-1\Longrightarrow\forall u,v\in X,\ \theta(pqu,a,v)&=E(p,a,-1,q)+\theta(u,a,v),\\
\eta(qu,a,pv)&=E(p,a,-1,q)+\theta(u,a,yv).\\
\end{split}
\label{eq.theta}
\end{equation}
Since $\Ac$ is $\delta$-local, for every triple $(u,a,v)$ in $X\times A\times X$,
if $\theta(u,a,v)$ is defined, it is uniquely defined.

For every vector $pu$ in $X$, $pu$ is initial if $p$ is in $\Iu$ and $(\varepsilon,\ltype,u)$ is in
$\underline{\eta}$; in this case, we set $\mathcal{I}(pu)=I(p)+\eta(\varepsilon,\ltype,u)$. Likewise, every vector $up$ in $X$ is final if $p$ is in $\Tu$ and $(u,\rtype,\varepsilon)$ is
in $\underline{\eta}$; in this case, we set $\mathcal{T}(up)=\eta(u,\rtype,\varepsilon)+T(p)$.
\begin{example}
For instance, with slices from automaton $\Ac_1$,
\begin{equation}
\begin{split}
\theta\left(\StackThreeLabels p{s_+}{q_+},a,\StackThreeLabels {q_-}rp\right)=&
E(p,a,1,q_-)+\eta\left(\StackTwoLabels {s_+}{q_+},a,\StackTwoLabels rp\right)\\
=&E(p,a,1,q_-)+E(r,a,-1,s_+)+\theta(q_+,a,p)\\
=&E(p,a,1,q_-)+E(r,a,-1,s_+)+E(q_+,a,1,p).\\
\end{split}
\end{equation}
The vector $\left[\StackThreeLabels p{s_+}{q_+}\right]$ is initial and
\begin{equation}
\mathcal{I}\left(\StackThreeLabels p{s_+}{q_+}\right)=I(p)+E(s_+,\ltype,1,q_+).
\end{equation}
\end{example}

\begin{definition}
Let $\Ac=(Q,A,E,I,T)$ be a two-way $\K$-automaton.
With the above notations,
the \emph{slice automaton} of $\Ac$ is
the infinite one-way $\K$-automaton~$\Sc=(X,A,\theta,\mathcal{I},\mathcal{T})$.
\end{definition}

\begin{proposition}\label{prop:slices}
Let $\K$ be a \emph{commutative} semiring and let $\Ac$ be a $\delta$-local two-way $\K$-automaton.
There is a bijection $\varphi$ between the computations of~$\Ac$ and the computations of
the slice automaton of~$\Ac$ such that, for every computation $\rho$ of $\Ac$,
\\-- $\rho$ and $\varphi(\rho)$ have the same label and the same weight;
\\-- the signature of $\rho$ is the sequence of states of $\varphi(\rho)$.
\end{proposition}
\begin{proof}
Let $\Sc$ be the slice automaton of~$\Ac$.
Let $\pi$ be a run in $\Sc$ with label $w$. Let $\pi^{(k)}$ be the prefix of length $k$ of $\pi$ and let $(v^{(0)},\ldots,v^{(k)})$ be
the sequence of states of $\pi^{(k)}$. We show by induction on $k$ that from $\pi^{(k)}$, there is a unique way to retrieve
the restriction of a run of $\Ac$ on $w$ to the $k$ first letters. Moreover, the weight of $\pi^{(k)}$ (including initial weight)
is equal to the weight of this restriction.
If $k=0$, $\pi^{(k)}$ is reduced to an initial slice.
By Equation~\ref{eq.theta}, the restriction of the path in the two-way automaton is uniquely defined: $v^{(0)}_1$ is initial with weight
$I(v^{(0)}_1)$, and for every $r$ in $[1;(v^{(0)}-1)/2]$, there is a transition $\Tr{v^{(0)}_{2r}}{\ltype,\rightarrow|h_r}{v^{(0)}_{2r+1}}$;
the weight of this restriction is actually the initial weight of $v^{(0)}$ in $\Sc$.
If $k>0$, we consider the restriction built for $k-1$; this restriction corresponds to a disjoint union of parts of the computations and there is only one way to
connect them to the states of the slice $v^{(k)}$ (since $\Ac$ is $\delta$-local). The weight of the transition between $v^{(k-1)}$ and $v^{(k)}$ is exacltly
the sum of the weights of the new transitions involved in the restriction.

Finally, from the restriction of length $|w|$, if we consider $v^{(|w|)}$ as a final state of $\Sc$, by an argument similar to the initial state, we obtain that
there is one and only one run in $\Ac$ that corresponds to a given run in $\Sc$.
\end{proof}
\subsection{Reduced computations and one-way automata}

In unweighted (or Boolean) automata, two-way automata describe exactly the same languages as one-way automata~\cite{Shepherdson59,RabinScott59}.
It is not always the case with weighted automata. For instance, let $\K$ be the semiring of languages of the alphabet $\{x,y\}$.
It is not difficult to design a deterministic two-way $\K$-automaton over the alphabet $\{a\}$
such that the image of $a^n$ is $x^ny^n$ (a first left-right traversal outputs an $x$ for each $a$,
then the automaton comes back to the beginning of the word and a second left-right traversal outputs a $y$ for each $a$). This function is obviously not rational and can not be realized by a one-way $\K$-automaton.

\begin{proposition}\label{cor:red-ow}
Let $\K$ be a \emph{commutative} semiring and let $\Ac$ be a $\delta$-local two-way $\K$-automaton.
There exists a (finite) one-way $\K$-automaton $\Bc$ such that
there is a bijection $\varphi$ between the reduced computations of~$\Ac$ and the computations
of~$\Bc$ such that, for every reduced computation $\rho$ of $\Ac$,
\\-- $\rho$ and $\varphi(\rho)$ have the same label and the same weight;
\\-- the signature of $\rho$ is the sequence of states of $\varphi(\rho)$.
\end{proposition}

\begin{proof}
Let $\Ac=(Q,A,E,I,T)$ be a two-way $\K$-automaton.
We consider vectors of elements of $Q$ such that no state of $Q$ appears twice at positions with the same parity. For all $k$ in $\N$, we set
\begin{equation}
\begin{split}
V_k=&\{v\in Q^{2k+1}\mid v_i=v_j\Rightarrow i\neq j\mod 2\}\\
=&\{v\in Q^{2k+1}\mid\forall p\in Q,\forall s\in[0;1], |\{i\mid v_i=p\text{ and  }
i=s\mod 2\}|\leqslant 1\}
\end{split}
\end{equation}
For every $k$ larger than $|Q|-1$, $V_k$ is empty. Let $V = \bigcup_kV_k$; we define
the one-way $\K$-automaton with set of states~$V$.
It is straightforward that a run is reduced if and only if every slice of this run is in $V$.

By Proposition~\ref{prop:slices}, the restriction of the slice automaton to $V$
gives a finite automaton that fulfils the proposition.
\end{proof}

Actually, the sufficient condition for the finiteness of the trim part of the slice automaton can be weaken.
If the number of slices of a two-way automaton is finite, it is equivalent to a one-way automaton. Unfortunately, this condition is not easy to check and is not a necessary condition.

\begin{example}
\begin{figure}
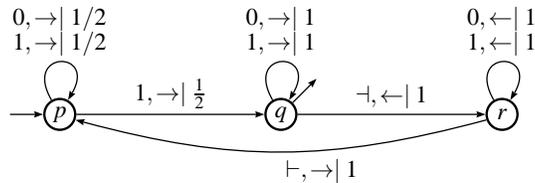
\FixVCScale{0.7}
\centering\VCCall{Dessins/twowayp.tex}
\caption{The two-way $\mathbb{Q}$-automaton $\Ac_2$.}
\end{figure}
The two-way $\mathbb{Q}$-automaton $\Ac_2$ computes for each word $w$ over the alphabet $\{0,1\}$ the value $\frac x{1-x}$, where $x=\sum_{i\in[1;|w|]}\frac{w_i}{2^i}$.
Although $\mathbb{Q}$ is commutative, this two-way automaton is not equivalent to any one-way $\mathbb{Q}$-automaton; $s_2=|\Ac_2|$ is not a rational series.
\end{example}





\begin{example}
Let $\Bc_1'$ be the trim part of the slice automaton of $\Ac_1'$ (Figure~\ref{fig:onewayc}).
In this particular case, although $\Ac_1$ is not $\delta$-local, the slice automaton $\Bc_1$ of $\Ac_1$
(Figure~\ref{fig:oneway}) is also unambiguous.
It has been shown in~\cite{KLMP04} that there is no deterministic one-way distance automaton
equivalent to these automata.

\begin{figure}\FixVCScale{0.7}
\centering\VCCall{Dessins/onewayq.tex}
\caption{The unambiguous one-way distance automata $\Bc_1$.}
\label{fig:onewayc}
\end{figure}
\begin{figure}\FixVCScale{0.7}
\centering\VCCall{Dessins/onewayp.tex}
\caption{The unambiguous one-way distance automata $\Bc_1$.}
\label{fig:oneway}
\end{figure}
\end{example}

\section{Unambiguity and Determinism}
Since every computation in an unambiguous two-way automaton is reduced,
Proposition~\ref{cor:red-ow} implies the following statement.

\begin{proposition}
Let $\K$ be a commutative semiring. Every unambiguous two-way $\K$-automaton is equivalent to
an unambiguous one-way $\K$-automaton.
\end{proposition}

A unambiguous one-way automaton can obviously be seen as a unambiguous two-way automaton.
In this part, we show that an unambiguous  one-way automaton can actually be converted into a determinstic two-way automaton.

\subsection{From Unambiguous one-way to Deterministic two-way Automata}

\begin{definition}\label{def:deterministic}
A two-way automaton is deterministic if
\\i) it has at most one initial state;
\\ii) for every state~$p$ and every letter~$a$, there is at most one transition outgoing from~$p$ with label~$a$;
\\iii) for every final state~$p$, there is no transition outgoing from~$p$ with label~$\rtype$.
\end{definition}

The last condition means that if a final state is reached at the end of the word, there is no nondeterministic choice between ending the computation and reading the right mark to continue.

\begin{theorem}\label{th:unambdet}
Let $\K$ be a semiring.
Every unambiguous one-way $\K$-automaton is equivalent to a deterministic two-way $\K$-automaton.
\end{theorem}
This result is an extension of~\cite{HopUll67}, where it is proved that an unambiguous one-way automaton can be simulated by a deterministic two-way automaton.
Our proof is inspired by~\cite{Carton12}, where it is proven that any rational function can be realized by a sequential two-way transducer.
Other works on the conversion of two-way transducers to one-way transducers can be found in~\cite{EngHoo07} or in~\cite{Souza13}.

\begin{proof}
Let $\Ac=(I,E,T)$ an unambiguous one-way $\K$-automaton with set of states $Q$.

We consider the mapping $\mu$ from $A$ into the $Q\times Q$ Boolean matrices defined by:
\begin{equation}
\forall a\in A,\ \forall p,q\in Q,\
\mu(a)_{p,q}=1 \Longleftrightarrow (p,a,q)\in\supp{E}.
\end{equation}
The monoid generated by $\{\mu(a)\mid a\in A\}$ is the \emph{transition} monoid $M$ of $\Ac$.
The mapping $\mu$ is naturally extended to a morphism of the monoid $A^*$ onto $M$.
Every subset of $Q$ can be interpreted as a vector in $\B^Q$; for every word $w$,
$\supp{I}\mu(w)$ is the set of states accessible from an initial state by a path with label $w$
and conversely, $\mu(w)\supp{T}$ is the set of states from which a terminal state can be reached by
a path with label $w$.

Since $\Ac$ is unambiguous, for every pair of words $(u,v)$, $\supp{I}\mu(u)\cap \mu(v)\supp{T}$
has at most one element (otherwise there would exist several computations accepting $uv$);
likewise, for every letter, there exists at most one transition $(p,a,q)$ in $\Ac$ with
$p$ in $\supp{I}\mu(u)$ and $q$ in $\mu(v)\supp{T}$ (otherwise there would exist several computations accepting $uav$).

For every word $w=w_1\dots w_k$, for every $i$ in $[0;k]$, we set
\begin{equation*}
X_i(w)=\supp{I}\mu(w_1\dots w_i)
\qquad\text{and}\qquad
Y_i(w)=\mu(w_{i+1}\dots w_k)\supp{T}.
\end{equation*}

We build a deterministic two-way $\K$-automaton $\Bc$ equivalent to $\Ac$.
$\Bc$ has the following property.
If $w$ is accepted by $\Bc$, for every $i$ in $[1;k]$,
the state reached after the last reading of $w_i$ contains the information $(X_i(w),Y_i(w))$:
\begin{center}
\VCCall{Dessins/path.tex}
\end{center}
From $(X_{i-1}(w),Y_{i-1}(w))$ and $(X_i(w),Y_i(w))$, the transition labeled by $w_i$ in the run
with label $w$ can be deduced: it is the only transition $(p,w_i,q)$ with $p$ in $X_{i-1}(w)$
and $q$ in $Y_i(w)$.
Likewise $(X_0,Y_0)$ determines the initial weight and $(X_k,Y_k)$ determines the final weight.

\medskip
The set $X_i$ can easily be deduced from $X_{i-1}$ : $X_i=X_{i-1}\mu(w_i)$.
the computation of $Y_i$ from $Y_{i-1}$ is more subtle.

Let $x$ and $y$ be two elements of $M$. If there exists $z$ in $M$ such that $x=zy$, we say that
$x\leqslant_Ly$; this relation is a preorder. If there also exists $t$ such that $tx=y$, we say
that $x$ and $y$ are L-equivalent.

Let $u$ be a factor of $w$ that starts in $w_{i+1}$. It obviously holds $\mu(w_iu) \leqslant_L\mu(u)$.
If $\mu(w_iu)$ and $\mu(u)$ are L-equivalent, there exists $y$ in M such that $y\mu(w_iu)=\mu(u)$.
In this case, it also holds $y\mu(w_i\dots w_k)=\mu(w_{i+1}\dots w_k)$ and therefore, $Y_i=yY_{i-1}$.
The two-way automaton can perform these computations, since they lie in the transition monoid, which is finite.
The automaton incrementally computes for each $j$ in $[i;k]$ the value of $\mu(w_{i+1}\dots w_j)$
until $\mu(w_i\dots w_j)<_L\mu(w_{i+1}\dots w_j)$ and $\mu(w_i\dots w_{j+1})\equiv_L\mu(w_{i+1}\dots w_{j+1})$.
If it reaches $j=k$, then $Y_i=\mu(w_{i+1}\dots w_j)T$, otherwise, $Y_i=yY_{i-1}$ where $y$ is such
that $y\mu(w_i\dots w_{j+1})=\mu(w_{i+1}\dots w_{j+1})$.

Once $Y_i$ is computed, the automaton must come back to position $i$. The automaton is in some position~$j$
such that $\mu(w_i\dots w_j)<_L\mu(w_{i+1}\dots w_j)$; {\it a fortiori}, for every $r$ in $[i+1;j]$,
$\mu(w_i\dots w_j)<_L\mu(w_r\dots w_j)$. The automaton therefore spans every position smaller than $j$
until it arrives to some point $s$ such that $\mu(w_i\dots w_j)=\mu(w_s\dots w_j)$. It then holds $s=i$.

Let $\mathcal{P}$ be the powerset of $Q$.
The set of states of $\Bc$ is the union of five kinds of states:
\\-- $Q_0=\{i\}$ is the initial state; in this state, the automaton read the input from left to right
     until it reached the right mark $\rtype$. It then goes to the state $\Tu$ in $Q_1$.
\\-- $Q_1\subseteq\mathcal{P}$; in this state, the automaton read the input $w$ from right to left;
     after reading the suffix $v$, the state corresponds to $\mu(v)\supp{T}$. When the left mark $\ltype$
	is reached, the automaton goes to the state $(\Iu,\mu(w)\Tu)$ in $Q_2$.
\\-- $Q_2\subseteq\mathcal{P}^2$; these states corresponds to the pairs $(X_i,Y_i)$; the incoming transitions on
     these states correspond to the transition of the one-way automaton; they are weighted by the corresponding
     weight. Likewise, a state in $Q_2$ may be terminal if it belongs to $\mathcal{P}\times\{T\}$.
     When the automaton is in one of these states, either it stops, or it starts to deal with a new letter;
     this letter is read and stored in the next state which belongs to $Q_3$.
\\-- $Q_3\subseteq A\times M\times \mathcal{P}^2$; the automaton stays in states $Q_3$ as long as it needs to
     compute $Y_i$ from $Y_{i-1}$. It stores the current letter $a$ as well as the image in the transition monoid of
     the factor $u$ that follows $a$ and ends at the current position.
     If the state stores $\mu(u)$ that is $L$-larger than $\mu(au)$ and the read letter $b$ is such that
     $\mu(aub)$ and $\mu(ub)$ are $L$-equivalent, there exists $y$ such that $y\mu(aub)=\mu(ub)$; then $Y_i=yY_{i-1}$, the automaton stores
     $\mu(au)$ and jump to a state in $Q_4$.
\\-- $Q_4\subseteq M^2\times \mathcal{P}^2$; the automaton stays in a state of $Q_4$ while it reads from right to left the word $u$;
     it stores the image of the suffix $v$ of $u$ which is read; it holds $\mu(v)>_L\mu(au)$ until $v=u$; at this point, the automaton read the
     letter $a$ and checks that $\mu(av)=\mu(au)$; at this point, it knows both $X_i$ and $Y_{i+1}$, therefore, it can output the weigth of the unique
     transition compatible with $a$, $X_i$ and $Y_{i+1}$, and jump to the state $(X_{i+1}=X_i\mu(a),Y_{i+1})$.
\begin{equation}
\begin{split}
F&=\{\Tr{i}{a,\rightarrow}{i\in Q_0}\mid a\in A\}\\
&\cup\{\Tr{i}{\rtype,\leftarrow}{\supp{T}\in Q_1}\}\\
&\cup\{\Tr{Y}{a,\leftarrow}{\mu(a)Y\in Q_1}\mid Y\in Q_1, a\in A\}\\
&\cup\{\Tr{Y}{\ltype,\rightarrow|I_k}{(I,Y)\in Q_2}\mid Y\in Q_1, k \in I\cap Y\}\\
&\cup\{\Tr{(X,Y)}{a,\rightarrow}{(a,\one,X,Y)\in Q_3}\mid (X,Y)\in Q_2, a\in A\}\\
&\cup\{\Tr{(a,x,X,Y)}{b,\rightarrow}{(a,x\mu(b),X,Y)\in Q_3}\mid (a,x,X,Y)\in Q_3, b\in A, \mu(a)x\mu(b)<_Lx\mu(b) \}\\
&\cup\{\Tr{(a,x,X,Y)}{b,\leftarrow}{(\mu(a)x,1,X,yY)\in Q_4}\mid (a,x,X,Y)\in Q_3, b\in A, y\in M,
y\mu(a)x\mu(b)=x\mu(b) \}\\
&\cup\{\Tr{(a,x,X,Y)}{\rtype,\leftarrow}{(\mu(a)x,1,X,T)\in Q_4}\mid (a,x,X,Y)\in Q_3 \}\\
&\cup\{\Tr{(x,y,X,Y)}{a,\leftarrow}{(x,\mu(a)y,X,Y)\in Q_4}\mid (x,y,X,Y)\in Q_4, a\in A, x<_L\mu(a)y \}\\\notag
&\cup\{\Tr{(x,y,X,Y)}{a,\rightarrow|k}{(X\mu(a),Y)\in Q_2}\mid (x,y,X,Y)\in Q_4, a\in A, x=\mu(a)y,\\
&\hspace*{8cm}\exists(p,q)\in X\times Y, \exists \Tr{p}{a|k}{q}\in\Ac \}.
\end{split}
\end{equation}
\end{proof}

For every $X$ in $\mathcal{P}$, if the state $(X,\supp{T})$ belongs to $Q_2$, $(X,\supp{T})$ is final
with weight $T_p$, where $p$ is the unique state in $X\cap\supp{T}$.

\begin{figure}[t!]
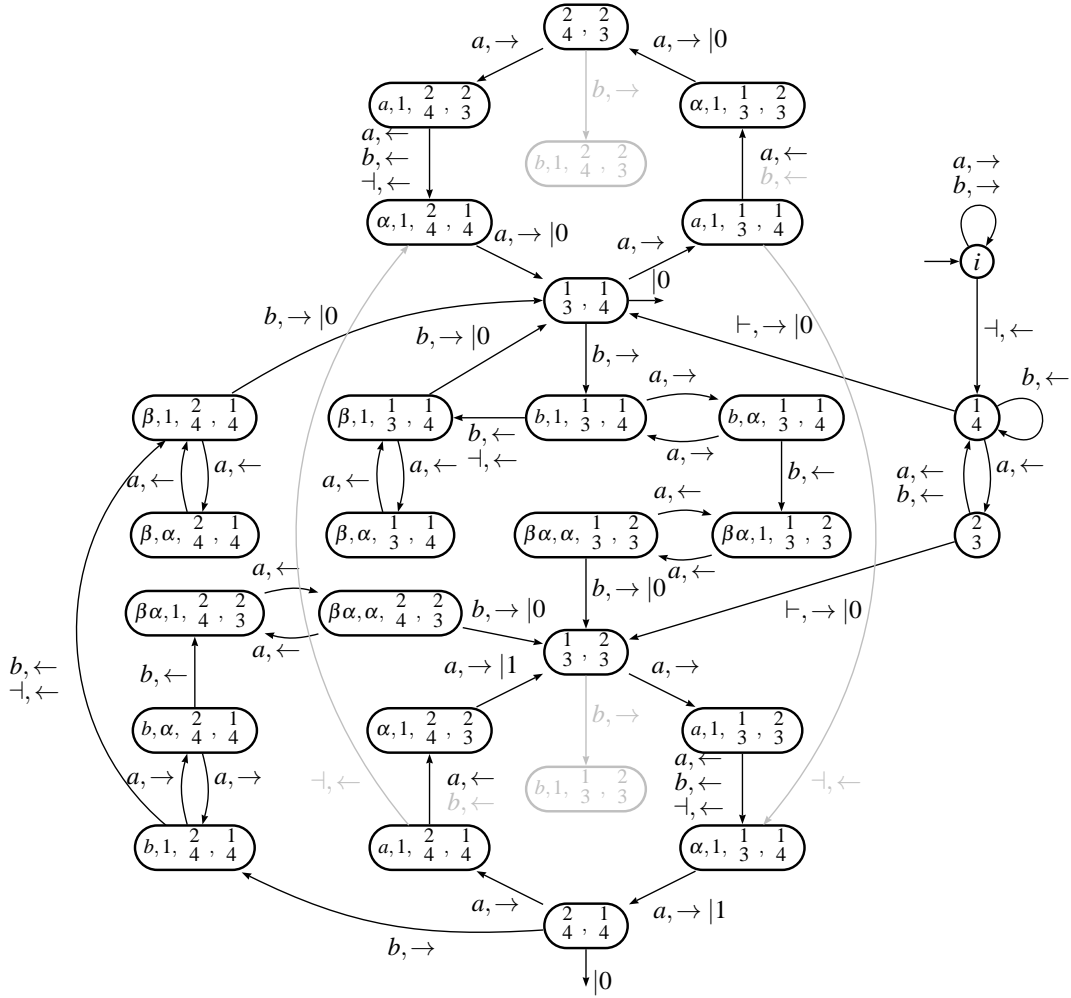
\FixVCScale{1}
\centering\VCCall{Dessins/unamb-det.tex}
\caption{The deterministic two-way distance automaton~$\Dc_1$.
The transitions or states in gray are not accessible. For sake of clearness, transitions outgoing from non accessible states are not drawn. Every column of numbers
is the set of non-zero components of a Boolean vector of size $4$.
The weights are only written on transitions where it comes from the weight of a transition (or from an initial/final
weight) of~$\Bc_1$.
}
\label{fig:unamb-det}
\end{figure}
\begin{example}
Let $\Bc_1$ be the unambiguous one-way automaton of Figure~\ref{fig:oneway}.
We number the states of this automaton: $[p]=1$, $[q]=2$, $[p,s,q]=3$ and $[q,r,p]=4$.
The transition monoid is generated by the following matrices:
\begin{equation}
\alpha=\mu(a)=\left[\begin{array}{cccc}
0 & 1 & 0 & 0\\
1 & 0 & 0 & 0\\
0 & 0 & 0 & 1\\
0 & 0 & 1 & 0\end{array}\right],\qquad
\beta=\mu(b)=\left[\begin{array}{cccc}
1 & 0 & 1 & 0\\
0 & 0 & 0 & 0\\
0 & 0 & 0 & 0\\
1 & 0 & 1 & 0\end{array}\right].
\end{equation}
The following identities hold : $\alpha^2=1$, $\beta^2=\beta$, $\beta\alpha\beta=\beta$.
It then holds $1\equiv_L\alpha$, $\alpha\beta\equiv_L\beta$ and $\alpha\beta\alpha\equiv_L\beta\alpha$,
while $\beta<_L1$ and $\beta\alpha<_L1$. Notice that $\beta$ and $\beta\alpha$ are uncomparable.
We can apply the proof of Theorem~\ref{th:unambdet} to compute the equivalent deterministic two-way automaton~$\Dc_1$ of Figure~\ref{fig:unamb-det}.
\end{example}

\begin{corollary}
Let $\K$ be a commutative semiring. Every unambiguous two-way $\K$-automaton is equivalent to a deterministic one.
\end{corollary}

\begin{remark}
This conversion can lead to a combinatorial blow-up. For instance, the deterministic two-way automaton
built from the unambiguous one-way automato~$\Bc_1$ (Figure~\ref{fig:oneway} (right)) has
27 states in its trim part.

A lower bound on the number of states can be computed. Let $n$ be the number of states of the unambiguous one-way automaton.
\begin{itemize}
\item $Q_0$ has one state;
\item $Q_1$ has at most $2^n-1$ states;
\item $Q_2$ is made of pairs of subset of $Q$ which share exactly one element, hence $Q_2$ has at most $n3^{n-1}$ states;
\item $Q_3$ is made of a pair of $Q_2$ endowed with a letter and an element of the transition monoid (that may have $2^{n^2}$ elements); hence $Q_3$ has at most $|A|n3^{n-1}2^{n^2}$ states;
\item $Q_4$ is made of two (non empty) subsets of $Q$ and two elements of the transition monoid; its size is bounded by $2^{2n+2n^2}$.
\end{itemize}
\end{remark}

\nocite{*}
\bibliographystyle{eptcs}
\bibliography{2wbiblio}

\end{document}